\documentclass[11pt]{article}

\usepackage{multirow}
\usepackage{color}
\usepackage{xspace}
\usepackage{ifpdf}
\ifpdf    
\usepackage{hyperref}
\else    
\usepackage[hypertex]{hyperref}
\fi
\usepackage{mathtools}
\usepackage{amsthm}
\usepackage{amsmath}
\usepackage{amssymb}
\usepackage{hyperref}
\usepackage{cleveref}
\usepackage{float}
\usepackage{graphicx}
\usepackage{xcolor}
\usepackage{caption}
\usepackage{clrscode}
\usepackage{fullpage}
\usepackage{thmtools}
\usepackage{tikz}

\newcommand{\Ot}{\widetilde{O}}

\graphicspath{ {./images/} }

\newtheorem{theorem}{Theorem}[section]
\newtheorem{lemma}[theorem]{Lemma}
\newtheorem{corollary}[theorem]{Corollary}
\newtheorem{claim}[theorem]{Claim}

\newtheorem{hypothesis}[theorem]{Hypothesis}

\newtheorem{observation}[theorem]{Observation}
\newtheorem{property}[theorem]{Property}

\newtheorem{problem}[theorem]{Problem}

\DeclarePairedDelimiter{\ceil}{\lceil}{\rceil}

\newcommand{\equalexpl}[1]{%
  \underset{\substack{\uparrow\\\mathrlap{\text{\hspace{-1em}#1}}}}{=}}
  
\newcommand{\leqexpl}[1]{%
  \underset{\substack{\uparrow\\\mathrlap{\text{\hspace{-1em}#1}}}}{\leq}}
  
  \newcommand{\geqexpl}[1]{%
  \underset{\substack{\uparrow\\\mathrlap{\text{\hspace{-1em}#1}}}}{\geq}}
  
\newcommand{\subseteqexpl}[1]{%
  \underset{\substack{\uparrow\\\mathrlap{\text{\hspace{-1em}#1}}}}{\subseteq}}

  \usepackage[inline]{enumitem}   
\makeatletter
\newcommand{\inlineitem}[1][]{%
\ifnum\enit@type=\tw@
{\descriptionlabel{#1}}
\hspace{\labelsep}%
\else
\ifnum\enit@type=\z@
\refstepcounter{@listctr}\fi
\quad\quad@itemlabel\hspace{\labelsep}%
\fi}
\makeatother

\newcommand\cfp{\hat{c}} 
\newcommand\cspace{\cfp}
\newcommand\nocdot{\,}
\newcommand\calpha{4}
\newcommand\calphabeta{16}

\newcommand\ccalphabeta{\calphabeta \cfp \nocdot }

\usepackage{authblk}

\title{On the Space Usage of Approximate Distance Oracles\\ with 
 Sub-2 Stretch\thanks{This work is partially supported by Israel Science Foundation (ISF) grant no. 1926/19 and  by United States - Israel Binational Science Foundation (BSF) grant no. 2018364.
Contact: \texttt{kopelot.biu@gmail.com,  korinariel10@gmail.com, 
liamr@macs.biu.ac.il}.}}

\author{Tsvi Kopelowitz, Ariel Korin, Liam Roditty \\Bar-Ilan University, Ramat Gan, Israel}

\date{}

\begin{document}

\maketitle
\thispagestyle{empty}
\setcounter{page}{0}
\begin{abstract}
For an undirected unweighted graph $G=(V,E)$ with $n$ vertices and $m$ edges, let $d(u,v)$ denote the distance from $u\in V$ to $v\in V$ in $G$.  
An $(\alpha,\beta)$-stretch approximate distance oracle (ADO) for $G$ is a data structure that given $u,v\in V$ returns in constant (or near constant) time a value $\hat d (u,v)$ such that $d(u,v) \le \hat d (u,v) \le \alpha\cdot d(u,v) + \beta$, for some reals $\alpha >1, \beta$.
If $\beta = 0$, we say that the ADO has stretch $\alpha$.

Thorup and Zwick~\cite{thorup2005approximate} showed that one cannot 
beat stretch 3 with subquadratic space (in terms of $n$) for general graphs.
P\v{a}tra\c{s}cu and Roditty~\cite{patrascu2010distance} showed that one can obtain stretch 2 using $O(m^{1/3}n^{4/3})$ space, and so if $m$ is subquadratic in $n$ then the space usage is also subquadratic. Moreover, P\v{a}tra\c{s}cu and Roditty~\cite{patrascu2010distance} showed that one cannot beat stretch 2 with subquadratic space even for  graphs where $m=\tilde{O}(n)$, based on the set-intersection hypothesis. 

In this paper, we investigate the minimum possible stretch achievable by an ADO as a function of the graph’s maximum degree, a study motivated by the question of identifying the conditions under which an ADO can be stored with subquadratic space while still ensuring a sub-2 stretch.

In particular, we show that if the maximum degree in $G$ is $\Delta_G \leq O(n^{1/k-\varepsilon})$ for some $0<\varepsilon \leq 1/k$, then there exists an ADO for $G$ that uses $\tilde{O}(n^{2-\frac {k\varepsilon}{3}})$ space and has a $(2,1-k)$-stretch. For $k=2$ this result implies a subquadratic sub-2 stretch ADO for graphs with $\Delta_G \leq O(n^{1/2-\varepsilon})$. 

We provide tight lower bounds for the upper bound under the same set intersection hypothesis, showing that if $\Delta_G = \Theta(n^{1/k})$, a $(2,1-k)$-stretch ADO requires $\tilde{\Omega}(n^2)$ space. Moreover, we show that for positive constants $\varepsilon, c > 0$, a $(2-\varepsilon, c)$-stretch ADO requires $\tilde{\Omega}(n^2)$ space even for graphs with $\Delta_G = \tilde{\Theta}(1)$.

\end{abstract}

\newpage


 
\section{Introduction}\label{sec:intro}
One of the most fundamental and classic problems in algorithmic research is the task of computing distances in graphs. 
Formally, given an undirected unweighted graph $G = (V,E)$, $|V| = n$ and $|E| = m$, the distance between two vertices $u, v \in V$, denoted $d(u,v)$, is the length of a shortest path  between $u$ and $v$.
A central problem in distance computations is the \emph{all-pairs shortest paths} (APSP) problem~\cite{dijkstra1959note, bellman1958routing, fredman1987fibonacci, cherkassky1996shortest, zwick2001exact, chan2010more, madkour2017survey} where the objective is to compute the distances between every pair of vertices in the graph. 
A main disadvantage in handling the output of the APSP problem is that  storing the distances between every pair of vertices in the graph requires $\Omega(n^2)$ space. 
As in many other problems in computer science, the lack of space efficiency in solving the APSP problem has motivated researchers to search for a tradeoff between space and accuracy. 
As a result, one central form of the APSP problem emerging from this line of research is constructing an \emph{approximate distance oracle} where we sacrifice accuracy for space efficiency.

\paragraph{Approximate Distance Oracles.}
An approximate distance oracle (ADO) is a space efficient data structure that produces distance estimations between any two vertices in the graph in constant or near-constant time.
Formally, given two vertices, $u, v \in V$, an ADO returns an estimation $\hat{d}(u,v)$ for the distance between $u$ and $v$ that satisfies: $d(u,v) \leq \hat{d}(u,v) \leq \alpha \cdot d(u,v)$, for some real $\alpha > 1$ which is called the \emph{stretch} of the ADO.
If the estimation of the ADO satisfies $d(u,v) \leq \hat{d}(u,v) \leq \max \{d(u,v), \, \alpha \cdot d(u,v)+\beta\}$ for some reals $\alpha > 1$ and $\beta$ (which can be negative), we say that the stretch of the ADO is an $(\alpha, \beta)$-stretch. 
\hfill

\paragraph{ADO for general graphs.}
ADOs were originally presented by Thorup and Zwick~\cite{thorup2005approximate} who designed a randomized algorithm that for any positive integer $k$ constructs an ADO for \emph{weighted} undirected graphs in $O(kmn^{1/k})$ time that uses $O(kn^{1+1/k})$ space and returns a $2k - 1$-stretch in $O(k)$ query time.

Thorup and Zwick~\cite{thorup2005approximate} showed that the space usage of their ADO construction for their given stretch is optimal for \emph{general graphs} based on the girth conjecture of Erd\H{o}s. 
Moreover, for stretch 3 (when $k=1$), the appropriate case of the girth conjecture is known to be true (due to complete bipartite graphs), and so the quadratic (in $n$) space lower bound for this case is unconditional.  
Notice that constructing an exact distance oracle in quadratic space is trivial. 

In the case where one allows for an additive error,  P\v{a}tra\c{s}cu and Roditty~\cite{patrascu2010distance} designed an algorithm which constructs a $(2,1)$-stretch ADO for unweighted graphs using $O(n^{5/3})$ space and $O(1)$ query time. Their result demonstrates that in such a case the multiplicative error can be reduced while still using subquadratic space.

\paragraph{ADO for sparse graphs.}
The (conditional) lower bound of Thorup and Zwick~\cite{thorup2005approximate} does not apply to sparser graphs with $m=o(n^{1+1/k})$, and indeed additional results show that it is possible to use subquadratic space and return a sub-3 stretch in such cases.
Specifically, P\v{a}tra\c{s}cu and Roditty~\cite{patrascu2010distance}, designed an algorithm that constructs a 2-stretch ADO using $O(m^{1/3} n^{4/3})$ space, and so for subquadratic $m$ the space usage is subquadratic.
P\v{a}tra\c{s}cu, Roditty and Thorup~\cite{patrascu2012new} presented additional tradeoffs for sub-3 stretch using subquadratic space for graphs where $m=\tilde O(n)$\footnote{Throughout this paper we use $\sim$ when suppressing poly-logarithmic factors in  asymptotic complexities.}.
Roddity and Tov~\cite{roditty2023approximate} improved the stretch of the ADO presented by Thorup and Zwick~\cite{thorup2005approximate} while using the same space for graphs with $m=\tilde O(n)$.

\paragraph{Conditional lower bounds and set-intersection.}
P\v{a}tra\c{s}cu and Roditty~\cite{patrascu2010distance} proved a lower bound for the space usage of sub-2 stretch ADOs (i.e., ADOs which satisfy $d(u,v) \leq \hat{d}(u,v) < 2 d(u,v)$) that holds even for sparse graphs, conditioned on the space usage for data structures that solve the following set intersection problem.
\begin{problem}\label{set intersection problem}
Let $X = \log^c N$ for a large enough constant $c$. Construct a data structure that preprocesses sets $S_1, \dots, S_N \subseteq [X]$, and answers queries of the form ``does $S_i$ intersect $S_j$?" in constant time.
\end{problem}

The lower bound proof by P\v{a}tra\c{s}cu and Roditty~\cite{patrascu2010distance} is based on the following hypothesis.
\begin{hypothesis}[\cite{patrascu2010distance,goldstein2017conditional,cohen2010hardness}]\label{set intersection hypothesis}
A data structure that solves Problem \ref{set intersection problem} requires $\tilde{\Omega}(N^2)$ space.
\end{hypothesis}

Since understanding the reduction by P\v{a}tra\c{s}cu and Roditty~\cite{patrascu2010distance} is useful for our results, we provide an overview of their reduction tailored to our needs. 
Given an instance of Problem~\ref{set intersection problem}, we construct a 3 \emph{layered} graph, where edges are only between adjacent layers, as follows. 
The first layer is $V_L =\{v_1,\ldots, v_N\}$, the second layer is $V_M = X$, and the third layer is $V_R =\{u_1,\ldots, u_N\}$.
Vertices $v_i$ and $u_i$ represent $S_i$, and so for each set $S_i$ and each $x\in S_i$, we add edges $(v_i,x)$ and $(x,u_i)$. 
Notice that the graph contains $\Theta(N)$ vertices.
It is straightforward to observe that $S_i\cap S_j\neq \emptyset$ if and only if there is a path of length 2 between $v_i$ and $u_j$. 
Moreover, since the graph is a 3 layered graph and the representatives of the sets are at the outer layers, there are no paths of length $3$ between representatives of sets.
Thus, one can solve Problem~\ref{set intersection problem} using a solution to the following problem (for $a=2$ and $b=4$).

\begin{problem}\label{def:distinguisher}
For positive integers $a$ and $b$, an \emph{$(a, b)$-distinguisher oracle} for a graph $G=(V, E)$, is a data structure that, given $u, v \in V$ establishes in constant time whether $d(u,v) \leq a$ or $d(u,v) \geq b$.  If $a < d(u,v) < b$ then the data structure can return any arbitrary answer.
\end{problem}

We conclude that a $(2,4)$-distinguisher oracle that  uses $f(n)$ space can be used to solve Problem~\ref{set intersection problem} using $f(N)$ space by applying the oracle onto the 3 layered graph.
Finally, since a sub-2 stretch ADO is a $(2,4)$-distinguisher oracle, any sub-2 stretch ADO must use at least $\Omega(n^2)$ space. 

\subsection{Main results: When can we beat stretch 2 with subquadratic space?}
The line of work by~\cite{patrascu2010distance,patrascu2012new,roditty2023approximate} is a natural research path given the observation that the (conditional) lower bounds of~\cite{thorup2005approximate} apply only  to graphs with $m=\Omega(n^{1+1/k})$.
Similarly, a natural goal, which we address in this paper, is to understand for which families of graphs can an ADO beat stretch $2$ using subquadratic space. 
In particular, the conditional lower bound proof of P\v{a}tra\c{s}cu and Roditty~\cite{patrascu2010distance} does not apply to graphs with  maximum degree of $n^{\frac 12 -\Omega(1)}$, since in such graphs the number of paths of length 2 is $n^{2-\Omega(1)}$, and so  constructing a subquadratic space $(2,4)$-distinguisher oracle is straightforward (by explicitly storing all length 2 paths). 

Thus, a natural goal, which we investigate in this paper, is to understand the relationship between the maximum degree of $G$, denoted by $\Delta_G$, and the best possible stretch obtainable for an ADO using subquadratic space. 
To address this question, we present an upper bound and matching lower bounds. The upper bound considers $\Delta_G = n^{\frac{1}{k}-\Omega(1)}$ and is summarized in the following theorem.

\begin{theorem}\label{thm:main_degree}
For any graph $G$, positive real constant $c$, and positive integer $k$ for which $\Delta_G \leq c n^{\frac{1}{k}-\varepsilon}$ for some real $0 < \varepsilon \leq 1/k$, there exists an ADO for $G$ that uses $\tilde{O}(c^k n^{2-\frac{k \varepsilon}{3}})$ space and has a $(2, 1-k)$-stretch.
\end{theorem}

For $k=2$, Theorem~\ref{thm:main_degree} implies a subquadratic sub-2 stretch ADO for graphs for which $\Delta_G \leq n^{\frac{1}{2}-\Omega(1)}$.

Next, we provide a tight conditional lower bound, conditioned on Hypothesis~\ref{set intersection hypothesis}, that applies when $\Delta_G = \Theta(n^{1/k})$, for all integers $k\ge 1$ \footnote{The case $k=1$ was proven by Thorup and Zwick~\cite{thorup2005approximate} to hold unconditionally. Thus, Theorem \ref{thm:main lower bound} focuses on $k \geq 2$.}. 
 The conditional lower bound is summarized  in the following Theorem.

\begin{theorem}\label{thm:main lower bound}
Let $2\le k \leq \log n$. Assuming Hypothesis \ref{set intersection hypothesis}, a $(2, 1-k)$-stretch ADO for graphs with $n$ vertices and maximum degree $\Theta(n^{1/k})$ must use $\tilde{\Omega}(n^2)$ space.
\end{theorem}

Notice that Theorem~\ref{thm:main lower bound} implies that the upper bound of Theorem~\ref{thm:main lower bound} is optimal under Hypothesis \ref{set intersection hypothesis} in two ways. First, while a subquadratic space $(2, 1-k)$-stretch ADO is achievable for $\Delta_G \leq O(n^{\frac{1}{k}-\varepsilon})$, achieving the same stretch for graphs with $\Delta_G = \Theta(n^{1/k})$ requires quadratic space. Secondly, even a slight improvement to the additive error of the ADO in the upper bound, such as a $(2, -k)$-stretch subquadratic ADO for graphs with $\Delta_G \leq O(n^{\frac{1}{k}-\varepsilon})$, would refute Hypothesis \ref{set intersection hypothesis} according to Theorem~\ref{thm:main lower bound}, since by setting $\varepsilon = \frac{1}{k} - \frac{1}{k+1}$, we obtain a subquadratic $(2, -k)$-stretch ADO for graphs with $\Delta_G \leq O(n^{\frac{1}{k}-\varepsilon}) = O(n^{\frac{1}{k+1}})$ which contradicts Theorem~\ref{thm:main lower bound}.

The tight bounds formed by Theorems \ref{thm:main_degree} and \ref{thm:main lower bound} leave open the question of whether it is possible to achieve an ADO with a $(2 - \varepsilon, c)$ stretch and subquadratic space for positive constants $\varepsilon, c$. In the following theorem we rule out the possibility of such an ADO assuming Hypothesis \ref{set intersection hypothesis}.

\begin{theorem}\label{thm:mult lower bound}
Let $\varepsilon, c > 0$ be constants. Assuming Hypothesis \ref{set intersection hypothesis}, a $(2 - \varepsilon, c)$-stretch ADO for graphs with $n$ vertices and maximum degree $\tilde{\Theta}(1)$ must use $\tilde{\Omega}(n^2)$ space.
\end{theorem}

Notice that Theorem~\ref{thm:mult lower bound} also holds for graphs with a maximum degree larger than $\tilde{\Theta}(1)$ since one could always add an isolated star component of at most $n$ vertices to a graph with $\Delta_G = \tilde{\Theta}(1)$ to make $\Delta_G$ arbitrarily large. Thus, if an ADO achieves certain bounds for graphs with a large $\Delta_G$ it could also match those bounds for graphs with $\Delta_G = \tilde{\Theta}(1)$.

\subsection{Organization}
The rest of this paper is organized as follows. 
In \Cref{sec:overview} we provide an overview of the main ideas used in this paper. 
In~\Cref{sec:additional related work} we survey some additional related work.
In~\Cref{sec:preliminaries} we provide some definitions that are used in the more technical parts of the paper. 
In~\Cref{sec:lemmas} we prove some useful lemmas that are used in the proof of~\Cref{thm:main_degree}, which is described in~\Cref{sec:A new ADO} together with the construction of our new ADO.
In~\Cref{sec: lower bound proof} we turn to the lower bounds and prove~\Cref{thm:main lower bound} and~\Cref{thm:mult lower bound}.
We conclude our work in~\Cref{sec:conclusions}. 

\subsection{Overview of Results and Techniques}\label{sec:overview}
In this section we describe an overview of the intuition and techniques used to obtain our  main results.

\subsubsection{Upper Bound: A New ADO}\label{sec:upper bound overview}

Since our new ADO is based on the ADO  of Agarwal and Godfrey~\cite{agarwal2013brief},  that has a $(2,1)$-stretch and uses $\tilde{O}(n^{5/3})$ space (which simplifies the ADO of~\cite{patrascu2010distance}), we provide an overview of the construction of their ADO.

The ADO constructed by Agarwal and Godfrey~\cite{agarwal2013brief} uses the concept of \emph{bunches} and \emph{clusters} introduced by Thorup and Zwick~\cite{thorup2001compact}. Following the conventions of Thorup and Zwick~\cite{thorup2001compact, thorup2005approximate}, for a vertex $v \in V$ and set $A \subseteq V$, let $p_A(v)$ be the vertex in $A$ which is closest to $v$ (breaking ties arbitrarily).
The \emph{bunch} $B_A(v)$ of $v$ with respect to $A$ is defined as $B_A(v) = \{ w \in V \mid d(v,w) < d(v,p_A(v)) \}$.
The \emph{cluster} $C_A(w)$ of $w$ with respect to $A$ is defined as $C_A(w) = \{ v \in V \mid d(w,v) < d(v,p_A(v)) \}$. 
We omit $A$ from the notation when it is clear from context. 
Thorup and Zwick~\cite{thorup2001compact} presented an algorithm that computes a set $A$ of size $\tilde{O}(s)$ such that $|B(v)|, |C(v)| \leq \calpha n/s$,  for every $v \in V$. 
The ADO of Agarwal and Godfrey~\cite{agarwal2013brief}  uses a hitting set $A$ of size $\tilde{O}(n^{2/3})$ such that for every $v \in V$, $|B(v)|, |C(v)| \leq O(n^{1/3})$.

Given two vertices $u,v\in V$, the ADO first tests whether $B(u)\cap B(v) \ne \emptyset$, and, if so, then the ADO returns the exact distance $d(u,v)$. 
The method for testing whether the two bunches intersect is based on the observation (which follows from the definitions of bunch and cluster) that $B(u)\cap B(v) \ne \emptyset$ if and only if $u\in C(B(v))$\footnote{For $S\subseteq V$, let $C(S)=\bigcup\limits_{u \in S}C(u)$.}.
Thus, each vertex $v\in V$ stores the exact distances to all vertices in $C(B(v))$, and now,  the case of $B(u)\cap B(v) \ne \emptyset$ costs constant time and returns an exact distance.
To deal with the case of $B(u)\cap B(v) = \emptyset$, the oracle stores the distances of pairs in  $V
\times A$, and the ADO returns the minimum of either the length of the shortest path between $u$ and $v$ passing through $p(u)$ or the length of the shortest path between $u$ and $v$  passing through $p(v)$.
The space usage is $O(n^{5/3})$ for storing $C(B(v))$ for every $v\in V$, and $\tilde O(n|A|) = \tilde O(n^{5/3})$ for storing the distances for all pairs in $V\times A$.

\paragraph{Intuition for the new ADO.}
Our new ADO construction is based on the following intuition  regarding the ADO construction of Agarwal and Godfrey~\cite{agarwal2013brief}. 
If we enlarge $B(u)$ by moving $p(u)$ to a further vertex (from $u$), then we would increase the likelihood of $B(u)\cap B(v) \ne \emptyset$, and so the ADO would return exact distances for more pairs of vertices.
However, in such a case, the quality guarantee on the stretch obtained by approximating $d(u,v)$ with  the shortest path from $u$ to $v$ that passes through $p(u)$ becomes worse. 
Part of the challenge is to balance the size of $B(u)$ which affects the usefulness of the intersections and the role of $p(u)$ when approximating the distances.

Our approach, intuitively, is to separate the definition of $p(u)$ used for the approximations and the set chosen for the intersections. 
Specifically, the definition of $p(u)$ remains unchanged relative to $A$ (we do however change the size of $A$), but instead of testing whether $B(u)\cap B(v)\neq \emptyset$, we use a larger set $N(u)$ (which contains $B(u)$), and test whether $N(u)\cap B(v)\neq \emptyset$ (or $B(u)\cap N(v)\neq \emptyset$).
Testing whether $N(u)\cap B(v)\neq \emptyset$ is implemented by storing all of the distances between $u$ and vertices in $C(N(u))$. 
We remark that one may consider the possibility of testing whether $N(u)\cap N(v)\neq \emptyset$ instead of testing whether $N(u)\cap B(v)\neq \emptyset$, however, such an approach seems to require too much space.

Recall that Agarwal and Godfrey~\cite{agarwal2013brief} obtained a $(2,1)$-stretch. 
In our algorithm, we choose $N(u)$ in such a way that when $N(u)\cap B(v) = \emptyset$ then $\min_{x \in B(u), y \in B(v)}\{ d(x, y) \} > k$ if $\Delta_G \leq n^{\frac{1}{k}-\Omega(1)}$, which ends up reducing the additive component of the stretch by at least $k$ (see \Cref{clm:stretch analysis}).
Thus, the approximation of the ADO is always at most $(2,1-k)$, which is less than stretch 2 for $k \geq 2$.

\subsubsection{Conditional Lower bound}
In Section~\ref{sec: lower bound proof} we  prove the following lemma, which directly implies Theorem~\ref{thm:main lower bound} since a $(2, 1-k)$-stretch ADO is also a $(k,k+2)$-distinguisher oracle.

\begin{lemma}\label{lem:CLB general distinguisher}
Let $2\le k \leq \log n$. Assuming Hypothesis \ref{set intersection hypothesis}, any $(k,k+2)$-distinguisher oracle for graphs with $n$ vertices and maximum degree $\Theta(n^{1/k})$ must use $\tilde{\Omega}(n^2)$ space.
\end{lemma}

Notice that it is straightforward to construct a $(k,k+2)$-distinguisher oracle for graphs with $n$ vertices and maximum degree $\Theta(n^{\frac 1k -\varepsilon})$ in $O (n^{2- k \varepsilon})$ space by storing all pairs of vertices at distance exactly $k$.
Thus, \Cref{lem:CLB general distinguisher} shows that such a construction is essentially optimal. 

\paragraph{The challenges.}
There are two issues that need to be addressed in order to extend the reduction by P\v{a}tra\c{s}cu and Roditty~\cite{patrascu2010distance} in a way that proves Lemma~\ref{lem:CLB general distinguisher}.  The first is to adjust the distances so that   $S_i\cap S_j\neq \emptyset$ if and only if $d(v_i,u_j) = k$, and otherwise, $d(v_i,u_j) \ge k+2$.
The second issue is that the degrees of vertices in $V_M$ need to be adjusted to be at most $\tilde O(N^{1/k})$.
In order to simplify our intuitive explanation, we focus our attention to the special case where $X=\{x\}$ has only one element.

One straightforward way of dealing with the first issue is to replace vertex $x\in V_M$ with a path $P_x=(w_1,\ldots,w_{k-1})$ of length $k-2$, and for each $S_i$ that contains $x$ we add edges $(v_i,w_1)$ and $(w_{k-1},u_i)$. 
Thus, the constructed graph would be a $k+1$ layered graph.
The number of vertices in such a graph is $O(k + N) = O(N)$, and the distance between vertices in the first and last layers are $k+2q$ for some integer $q\ge 0$. 
However, we still need to address the second issue of bounding the maximum degree, since $w_1$ and $w_{k-1}$ may have a very high degree corresponding to the number of sets containing $x$.

On the other hand, one initial idea (that does not work) for dealing with the second issue is to replace $x\in V_M$ (in the original 3 layered graph) with $N$ vertices $y_1,y_2,\ldots,y_N$,
and for each $S_i$ that contains $x$ we add edges $(v_i,y_i)$ and $(y_{i},u_i)$. 
Now the maximum degree of each node is constant, however, for $i\ne j$ such that $x\in S_i\cap S_j$, there is no path from $v_i$ to $u_j$. 
This idea is missing the functionality of the path $P_x$ which allows us to connect more than one pair of vertices from $V_L \times V_R$.

\paragraph{Combining approaches.}
Our reduction makes use of an underlying $k+1$ layered \emph{infrastructure graph} $\mathcal L$, commonly known as the \emph{butterfly graph} (see~\cite{p2011unifying, patrascu2012new}), which has the following three properties: 
\begin{enumerate*}
    \item each layer contains $N$ vertices, 
    \item there is a path of length $k$ from every vertex in the first layer to every vertex in the last layer, and 
    \item the degree of every vertex is at most $2N^{1/k}$.
\end{enumerate*}
The layers of $\mathcal L$ are numbered $0$ to $k$.
The vertices in each layer  are (separately) indexed with integers from $1$ to $N$, and the construction of $\mathcal L$ is based on the base $N^{1/k}$ representation of the these indices: for $1\le t \le k$, vertices from layer $t-1$ are connected with vertices from layer $t$ if and only if the base $N^{1/k}$ representation of their corresponding indices are the same, except for possibly the $t$'th digit. 
Similar to before, we denote the first layer of $\mathcal{L}$ by $V_L =\{v_1,\ldots, v_N\}$ and the last layer by $V_R =\{u_1,\ldots, u_N\}$.

Finally, we construct a $k+1$ layered graph $G_x$ which is intuitively obtained by removing from $\mathcal L$ edges touching either $v_i$ or $u_i$ for every $S_i$ that does not contain $x$.  
Thus, in $G_x$, if $x\in S_i\cap S_j$ then there is a path of length $k$ from $v_i$ to $u_j$ in $G_x$, and otherwise, there is no path from $v_i$ to $u_j$ in $G_x$.

We remark that in the general case, where $|X|$ may be larger than 1,
we combine $G_x$ for different $x\in X$ in a special way, and so we may introduce paths from $v_i$ to $u_j$ even if $S_i\cap S_j = \emptyset$. However, since the resulting graph is still a $k+1$ layered graph, and we are interested in paths between vertices in the first layer and vertices in the last layer, the lengths of such paths must be at least $k+2$.
Thus, a $(k,k+2)$-distinguisher oracle on the combined graph suffices for solving \Cref{set intersection problem}.

\paragraph{Applying the edge splitting technique.}
In order to prove Theorem~\ref{thm:mult lower bound}, which assuming Hypothesis \ref{set intersection hypothesis} eliminates the possibility of a subquadratic $(2 - \varepsilon, c)$-stretch ADO, even for graphs with $\Delta_G = \tilde{\Theta}(1)$, we use a construction which is based on the construction we use in the proof of  Theorem~\ref{thm:main lower bound}, but with two changes: 
\begin{enumerate*}
    \item we use $k = \log(n)$ so we have $\Delta_G = \tilde{O}(N^{1/k}) = \tilde{O}(1)$,
    \item we use the \emph{edge splitting} technique also used in \cite{patrascu2010distance,cohen2010hardness}, but unlike in \cite{patrascu2010distance,cohen2010hardness} we only split certain edges in the graph so the lower bound on the stretch is as high as possible.
\end{enumerate*}
See~\Cref{sec: lower bound proof} for more details.

\subsection{Additional related work}\label{sec:additional related work}

Different aspects of Thorup and Zwick~\cite{thorup2005approximate} ADOs were studied since they were introduced for the first time.
Chechik~\cite{Chechik15,Chechik14} reduced  the query time from $O(k)$ to $O(1)$,  while keeping the stretch and the space unchanged. (See also~\cite{Nilsen13,MeNa06}.)
Roditty, Thorup, and Zwick \cite{roditty2005deterministic} presented a deterministic algorithm that constructs an ADO in $\Ot(mn^{1/k})$ time while keeping the stretch and the space unchanged. 
Baswana and Kavitha \cite{baswana2006faster} presented an algorithm  with $O(n^2 \log n)$ running time\footnote{For $k=2$ the query time is $O(\log n)$. For $k>2$ the query time is $O(k)$.}. Baswana, Goyaland and Sen~\cite{BaswanaGS09}  presented an $\Ot(n^2)$ time algorithm that computes a $(2,3)$-distance oracle with $\Ot(n^{5/3})$ space. 
Sommer~\cite{Sommer16} presented an $\Ot(n^2)$ time algorithm that computes a $(2,1)$-distance oracle with $\Ot(n^{5/3})$ space. 
Akav and Roditty~\cite{akav2020almost} presented the first sub-quadratic time algorithm that constructs an ADO with stretch better than $3$. They presented an $O(n^{2-\epsilon})$-time algorithm that constructs a ADO with $O(n^{11/6})$ space and $(2+\epsilon,5)$-stretch. 
Chechik and Zhang~\cite{ChechikZ22} improved the result of Akav and Roditty~\cite{akav2020almost}. 
Among their results is an $O(m+n^{1.987})$ time algorithm that constructs an ADO with $(2,3)$-stretch and  $\Ot(n^{5/3})$ space.
Following the work by P\v{a}tra\c{s}cu and Roditty \cite{patrascu2010distance} who constructed an ADO for unweighted graphs that uses $O(n^{5/3})$ space and returns a $(2,1)$-stretch in $O(1)$ time, Abraham and Gavoille \cite{abraham2011approximate} extended the ADO by P\v{a}tra\c{s}cu and Roditty \cite{patrascu2010distance} for all even stretch values, by constructing for any integer $k \geq 2$, an ADO of size $\tilde{O}(n^{1+2/(2k-1)})$ with a $(2k - 2,1)$-stretch returned in $O(k)$ time. P\v{a}tra\c{s}cu, Roditty and Thorup \cite{patrascu2012new} focused on analyzing sparse graphs where $m = \tilde{O}(n)$ and noted that both the ADOs by Thorup and Zwick \cite{thorup2005approximate}, and the ADOs by Abraham and Gavoille \cite{abraham2011approximate} use a space complexity that can be described by the curve $S(\alpha, m) = \tilde{O}(m^{1+2/(\alpha+1)})$ where $\alpha$ is the stretch of the ADO and $m$ is the number of edges in the graph. 
P\v{a}tra\c{s}cu, Roditty and Thorup \cite{patrascu2012new} extended the curve $S(\alpha, m)$ to work for non integer stretch values $\alpha > 2$.
Although our research focuses on constant query time ADOs, another branch of research includes ADOs that have non constant query time~\cite{agarwal2011approximate, porat2013preprocess, agarwal2013distance, agarwal2014space, bilo2023improved}. 

In the lower bound regime, the problem of constructing a $(2,4)$-distinguisher oracle was analyzed from the perspective of time complexity as well. 
For graphs with degree of at most $n^{1/2}$, the problem of determining for each edge in the graph whether it is in a triangle in $O(n^{2-\varepsilon})$ time for some $\varepsilon > 0$ was shown to be hard under either the $\mathsf{3SUM}$ \cite{patrascu2010towards, kopelowitz2016higher} or $\mathsf{APSP}$ \cite{williams2020monochromatic} hypotheses. 
Since there exists a standard reduction from the problem of determining for each edge in the graph whether it is in a triangle to the problem of constructing a $(2,4)$-distinguisher oracle (see~\cite{abboud2022hardness}), a $(2,4)$-distinguisher oracle is also hard to construct in subquadratic time for graphs with degree of at most $n^{1/2}$ under either the $\mathsf{3SUM}$ or $\mathsf{APSP}$ hypotheses.

The problem of constructing a $(k,k+2)$-distinguisher oracle for a general integer $k  \geq 2$ was also studied in the past in terms of time complexity. 
Dor, Halperin and Zwick~\cite{dor2000all} showed that if all distances in an undirected $n$ vertex graph can be approximated with an additive error of at most $1$ in $O(A(n))$ time, then \emph{Boolean matrix multiplication} on matrices of size $n \times n$ can also be performed in $O(A(n))$ time. 
Dor, Halperin and Zwick~\cite{dor2000all} conclude that constructing a $(k,k+2)$-distinguisher oracle for an integer $k \geq 2$ is at least as hard as multiplying two Boolean matrices.

\section{Preliminaries}\label{sec:preliminaries}

Let $d_G(u, v)$ be the distance between vertices $u$ and $v$ in the graph $G$. The eccentricity of a vertex $v \in V$ in a graph $G$, denoted by $ecc_G(v)$, is defined as $ecc_G(v) = \max_{u \in V}\{d_G(v, u)\}$. 
The diameter of $G$ is defined as $diam_G = \max_{v \in V}\{ecc_G(v)\}$ and the radius of $G$ is defined as $rad_G = \min_{v \in V}\{ecc_G(v)\}$. 

The eccentricity of a vertex $v$ can be thought of as the distance between $v$ and the last vertex met during a \emph{Breadth First Search} (BFS) of the graph starting at $v$.
Since our goal is to construct an ADO that uses subquadratic space, we cannot afford to store a separate BFS tree for each vertex.
Instead, the construction algorithm of the ADO from Theorem~\ref{thm:main_degree} will store only a partial BFS tree for each vertex by truncating the BFS scan after some number of vertices. 
Motivated by this notion of a truncated  scan, we introduce the following generalization of eccentricity which turns out to be useful for our purposes.


Let $N_G(v, s)$ be the first $s$ vertices met during a BFS\footnote{ The  traversal order of vertices in the same layer during the BFS execution does not matter as long as the order is consistent. } starting from $v$ in the graph $G$, i.e., the $s$ closest vertices to $v$ (excluding $v$). 
If $s$ is not an integer, then let $N(v, s) = N(v, \lfloor s \rfloor)$. 
For an integer $r \geq 0$, define $L_G(v, r) = \{ u \in V \setminus\{v\} \mid d_G(u, v) = r \}$ and $T_G(v, r) = \{ u \in V \mid 0 < d_G(u, v) \leq r \}$. 
Notice that $T_G(v, r) = \bigcup\limits_{i = 1}^{r}L_G(v, i)$. 
For any real $1 \leq s \leq n - 1$, define $ecc_G(v, s)$ to be the maximum integer $k \in \left[0, ecc_G(v) \right]$ for which $T_G(v, k) \subseteq N_G(v, s)$.  
Notice that $ecc_G(v, n-1) = ecc_G(v)$. 
Define 
$rad_G(s) = \min_{v \in V}\{ecc_G(v, s)\}$. 
Notice that $rad_G(n-1) = rad_G$. 
We omit the subscript $G$ when using the definitions above whenever $G$ is clear from context.

\section{Useful Lemmas}\label{sec:lemmas}
In this section we prove several useful properties of the graph attributes defined in Section \ref{sec:preliminaries} which will be used to prove the upper bound theorem in~\Cref{sec:A new ADO}.

The following observation and corollary address the relationship between $T(v,r)$ and $N(v,s)$, and follow from the definition of BFS.

\begin{observation}\label{obs:1}
Let $G = (V,E)$  be an unweighted undirected graph, with $|V| = n$. For any $v \in V$ and integers $s$ and $r$ such that $1 \leq s < n$ and $1 \leq r \leq diam_G$, either 
\begin{enumerate*} 
\item $T(v, r) \subset N(v, s)$, 
\item $N(v, s) \subset T(v, r)$, or 
\item $N(v, s) = T(v, r)$ 
\end{enumerate*}.
\end{observation}

\begin{corollary}\label{cor:1}
Let $G=(V,E)$ be an unweighted undirected graph, with $|V| = n$.
For any $v \in V$ and integers $s$ and $r$ such that $1 \leq s < n$ and $1 \leq r \leq diam_G$, 
\begin{enumerate*} 
\item if $|T(v, r)| < |N(v,  s)|$ then $T(v, r) \subset N(v, s)$, 
\item if $|N(v,  s)| < |T(v, r)|$ then $N(v, s) \subset T(v, r)$, and 
\item if $|N(v,  s)| = |T(v, r)|$ then $N(v, s) = T(v, r)$\end{enumerate*}.
\end{corollary}

The following useful property addresses the relationship between $T(v,r)$ and $N(v,s)$ for the special cases where either $r=ecc(v,s)$ or $r=ecc(v,s)+1$. 

\begin{property}\label{prop:1}
Let $G=(V,E)$ be an unweighted undirected graph, with $|V| = n$. For any $v \in V$ and integer $s$ such that $1 \leq s < n$, we have: 
\begin{enumerate*} 
\item $T(v, ecc(v, s)) \subseteq N(v, s)$, and
\item  if $s < n - 1$, then $T(v, ecc(v, s) + 1) \not\subseteq N(v, s)$ 
\end{enumerate*}.\end{property}

\begin{proof} 
By definition, $ecc(v, s)$ is the largest integer $k \in \left[0, ecc(v) \right]$ for which $T(v, k) \subseteq N(v, s)$. 
Thus, 
\begin{enumerate*} 
\item $T(v, ecc(v, s)) \subseteq N(v, s)$, and 
\item if $ecc(v, s) < ecc(v)$ then $T(v, ecc(v, s) + 1) \not\subseteq N(v, s)$ 
\end{enumerate*}. 
If $s < n - 1$, it must be that $ecc(v, s) < ecc(v)$, since if we assume towards a contradiction that $ecc(v, s) = ecc(v)$ for some $s < n - 1$ then $T(v, ecc(v, s)) = T(v, ecc(v)) = V \setminus \{v\} = N(v, n-1)$ but on the other hand, by definition of $ecc(v, s)$, we have $T(v, ecc(v, s)) \subseteq N(v, s) \subset N(v, n-1)$, which is a contradiction.
\end{proof}

The following  lemma states that $ecc_G(v,s)$ exhibits a behavior that is similar to the behavior of the distance function which cannot decrease when removing edges  and vertices from $G$. 

\begin{lemma}\label{lem:1}
Let $G=(V,E)$ be an unweighted undirected graph, $V' \subseteq V$, and let $G'$ be the subgraph of $G$ induced by the vertices in $V'$. For any vertex $v \in V'$ and for any integer $s$ such that $1 \leq s < |V'|$, it holds that $ecc_G(v, s) \leq ecc_{G'}(v, s)$.
\end{lemma}

\begin{proof}
Given an integer $1 \leq s < |V'|$, let $r = ecc_{G}(v, s)$ and $r' = ecc_{G'}(v, s)$. We want to show that $r \leq r'$. By definition of $ecc_{G'}(v, s)$, $r' = ecc_{G'}(v, s)$ is the largest value for which $T_{G'}(v, r') \subseteq N_{G'}(v, s)$. 
Thus, in order to show that $r \leq r'$, it suffices to show that $T_{G'}(v, r) \subseteq N_{G'}(v, s)$.

For any vertex pair $u, w \in V'$, we have $d_G(u,w) \leq d_{G'}(u,w)$ since  $G'$ is a subgraph of $G$. Thus, \begin{align*}
T_{G'}(v, r) &= \{ u \in V' \mid 0 < d_{G'}(v, u) \leq r \} \\
&\subseteq \{ u \in V \mid 0 < d_{G}(v, u) \leq r \} \\
&= T_{G}(v, r) \\
&\subseteqexpl{Property \ref{prop:1}} N_{G}(v, s) .
\end{align*}
This implies that $|T_{G'}(v, r)| \leq |N_{G}(v, s)| = s = |N_{G'}(v, s)|$. 
By Corollary \ref{cor:1}, since $|T_{G'}(v, r)| \leq |N_{G'}(v, s)|$ then  $T_{G'}(v, r) \subseteq N_{G'}(v, s)$, as required. 
\end{proof}

\subsection{The Logarithmic-Like Behavior of Eccentricity} 
In the following lemma, which is an important ingredient in the analysis of our new ADO, we show that $ecc(v, s)$ satisfies a logarithmic-like behavior. Specifically,  $\log(x y) = \log x + \log y$. 
The reason for this behavior is that the number of vertices in each layer of a BFS tree expands in a similar way to an exponential function. 
For a tree-graph $G$ with minimum degree $\delta$ rooted at a vertex $v$,  for integers $0 \leq i < t < ecc(v)$ it holds that  $|L(v, t)| \geq \delta^i \cdot |L(v, t-i)|$. 
Since the number of vertices in every layer of the rooted tree grows exponentially, the eccentricity $ecc(v, s)$ grows logarithmically (in relation to $s$). 
Unlike in trees where the expansion of the number of vertices in every layer of a BFS can be analyzed using $\delta$, for general graphs, in order to achieve a lower bound for the expansion rate of the eccentricity of the vertices, we use $rad_G(s)$ instead.

\begin{lemma}\label{lem:2}
Let $G=(V,E)$ be an unweighted undirected graph, with $|V| = n$. 
For any vertex $v \in V$ and integers $s_1, s_2 \geq 1$ such that $s_1(s_2 + 1) < n - 1$, it holds that $ecc_G(v, s_1(s_2 + 1)) \geq ecc_G(v, s_1) + rad_G(s_2)$.
\end{lemma}
\begin{proof} 
Assume towards a contradiction that $ecc_G(v, s_1(s_2 + 1)) < ecc_G(v, s_1) + rad_G(s_2)$. Thus, $ecc_G(v, s_1(s_2 + 1)) + 1 \leq ecc_G(v, s_1) + rad_G(s_2)$. 

Let $T_{BFS}$ be a BFS tree rooted at $v$ in graph $G$. 
Let $\ell = |L(v, ecc_G(v, s_1))|$  and let $u_1, u_2, \dots , u_{\ell}$ be the vertices in $L(v, ecc_G(v, s_1))$. 
For any $i$, where $1 \leq i \leq \ell$, let $V_i$ be the set of descendant of $u_i$ in $T_{BFS}$ and let $G_i$ be the graph induced by $V_i$ in $G$\footnote{It is important to note that for all scans referenced in this proof, which include a BFS procedure of $G$ starting at $v$ and BFS procedures of $G_i$ starting at $u_i$ for $1 \leq i \leq \ell$, we require a consistent order of scanning, i.e., that for a given $1 \leq i \leq \ell$, and vertices $x, x' \in V_i \subseteq V$, if $x$ is scanned before $x'$ in $G$ then $x$ should also be scanned before $x'$ in $G_i$ (and vice versa). 
This is a valid requirement since for any vertices $y, y' \in V_i \subseteq V$, $d_G(v, y) \leq d_G(v, y')$ if and only if $d_{G_i}(u_i, y) \leq d_{G_i}(u_i, y')$.}. 

Let $\mu = \min_{i \in [1, \ell]}\{ecc_{G_i}(u_i, s_2)\}$. 
We will show that: \begin{align}\{ u \in V \mid ecc_G(v, s_1) < d_G(v, u) \leq ecc_G(v, s_1) +  \mu\} \subseteq \bigcup\limits_{i = 1}^{\ell}N_{G_i}(u_i, s_2). \label{eq:00}\end{align} 
Let $w \in \{ u \in V \mid ecc_G(v, s_1) < d_G(v, u) \leq ecc_G(v, s_1) +  \mu\}$. 
By the definition of BFS, since $d_G(v, w) > ecc_G(v, s_1)$, $w$ must be a descendant of some vertex $u_j$, and so $w \in V_j$.
Since $T_{BFS}$ is a shortest path tree rooted at $v$ and since $w \in V_j$, it must be that $d_G(v, w) = d_G(v, u_j) + d_G(u_j, w) = ecc_G(v, s_1) + d_G(u_j, w)$. 
By definition of $w$, $d_G(v, w) \leq ecc_G(v, s_1) +  \mu$. Thus, $ecc_G(v, s_1) + d_G(u_j, w) \leq ecc_G(v, s_1) +  \mu$ and so $d_G(u_j, w) \leq \mu$.

By definition, $T_{G_j}(u_j, ecc_{G_j}(u_j, s_2)) = \{ x \mid x \in V_j \, \wedge \, 0 < d_{G_j}(u_j, x) \leq ecc_{G_j}(u_j, s_2)\}$. 
Since $T_{BFS}$ is a shortest path tree, for any $y \in V_j$ it must be that $d_{G_j}(u_j, y) = d_{G}(u_j, y)$. Thus, $T_{G_j}(u_j, ecc_{G_j}(u_j, s_2)) = \{ x \mid x \in V_j \, \wedge \, 0 < d_{G}(u_j, x) \leq ecc_{G_j}(u_j, s_2)\}$. 

Since $w \in V_j$, and since $d_G(u_j, w) \leq \mu \leq ecc_{G_j}(u_j, s_2)$, then $w \in T_{G_j}(u_j, ecc_{G_j}(u_j, s_2))$. By Property \ref{prop:1}, $T_{G_j}(u_j, ecc_{G_j}(u_j, s_2)) \subseteq N_{G_j}(u_j, s_2)$. It follows that every vertex in $\{ u \in V \mid ecc_G(v, s_1) < d_G(v, u) \leq ecc_G(v, s_1) +  \mu\}$ must be included in $N_{G_i}(u_i, s_2)$ for some $i$, thus confirming Equation \eqref{eq:00}. 

By Property \ref{prop:1}, $\{ u \in V \mid 0 < d_G(v, u) \leq ecc_G(v, s_1) \} = T(v, ecc_G(v, s_1)) \subseteq N(v, s_1)$. Combining with Equation \eqref{eq:00} we have that $\{ u \in V \mid 0 < d_G(v, u) \leq ecc_G(v, s_1) \} \cup \{ u \in V \mid ecc_G(v, s_1) < d_G(v, u) \leq ecc_G(v, s_1) +  \mu\} \subseteq N(v, s_1) \cup \left( \bigcup\limits_{i = 1}^{\ell}N_{G_i}(u_i, s_2) \right)$, and so:
\begin{align*} \{ u \in V \mid 0 < d(v, u) \leq ecc_G(v, s_1) + \mu \} = T(v,  ecc_G(v, s_1) + \mu) \subseteq N(v, s_1) \cup \left( \bigcup\limits_{i = 1}^{\ell}N_{G_i}(u_i, s_2) \right).\end{align*} 

Now, 
\begin{align*}ecc_G(v, s_1) + \mu &\equalexpl{definition of $\mu$} ecc_G(v, s_1) + \min_{i \in [1, \ell]}\{ecc_{G_i}(u_i, s_2)\} \\
&\geqexpl{by Lemma \ref{lem:1}} ecc_G(v, s_1) + \min_{i \in [1, \ell]}\{ecc_{G}(u_i, s_2)\} \\
&\geq ecc_G(v, s_1) + \min_{u \in V}\{ecc_G(u, s_2)\} \\
&\equalexpl{definition of $rad_G(s_2)$} ecc_G(v, s_1) + rad_G(s_2) \\
&\geqexpl{assumption} ecc_G(v, s_1(s_2 + 1)) + 1. \end{align*}

Thus, $T(v, ecc_G(v, s_1(s_2 + 1)) + 1) \subseteq T(v, ecc_G(v, s_1) + \mu) \subseteq N(v, s_1) \cup \left( \bigcup\limits_{i = 1}^{\ell}N_{G_i}(u_i, s_2) \right)$. 

Notice that $\ell \leq s_1$, since, by Property \ref{prop:1}, $\ell = |L(v, ecc_G(v, s_1))| \leq |T(v, ecc_G(v, s_1))| \le N(v, s_1) = s_1$. 
Therefore,
\begin{align*}
|T(v, ecc_G(v, s_1(s_2 + 1)) + 1)| &\leq |N(v, s_1) \cup \left( \bigcup\limits_{i = 1}^{\ell}N_{G_i}(u_i, s_2) \right)| \\ 
&\leq |N(v, s_1)| + \sum\limits_{i=1}^{\ell}|N_{G_i}(u_i, s_2)| \\ 
&\leq s_1 + \ell \cdot s_2 \\
&\leq s_1 + s_1 \cdot s_2 \\
&= s_1(1 + s_2) \\
&= |N(v, s_1(s_2 + 1))|. 
\end{align*}
By Corollary \ref{cor:1}, it follows that $T(v, ecc_G(v, s_1(s_2 + 1)) + 1) \subseteq N(v, s_1(s_2 + 1))$, which contradicts Property \ref{prop:1}.
\end{proof}

\section{Upper Bounds}\label{sec:A new ADO}

In this section we prove Theorem~\ref{thm:main_degree} by introducing a new ADO which uses subquadratic space and produces a $(2, 1-k)$-stretch for graphs for which $\Delta_G \leq O(n^{1/k-\varepsilon})$ for a positive integer $k$ and real constant $0 < \varepsilon \leq 1/k$.
The ADO is parameterized by a parameter $0 \leq \alpha < 1/3$ which quantifies the tradeoff between the space and the stretch of the ADO. 
When $\alpha = 0$ the ADO is very similar to the ADO of Agarwal and Godfrey~\cite{agarwal2013brief} which uses $\tilde{O}(n^{5/3})$ space and has a $(2,1)$-stretch. For $0 < \alpha < 1/3$, the ADO uses additional space and is able to improve the stretch of the ADO for the family of graphs for which $\Delta_G \leq O(n^{3 \alpha / k})$ .

\subsection{The Construction Algorithm} 

The description of our construction algorithm follows the notations and definitions described in Section~\ref{sec:upper bound overview}.
The construction begins with an algorithm of Thorup and Zwick~\cite{thorup2001compact} that computes a set $A$ of size $\tilde{O}(s)$ such that $|B(v)|, |C(v)| \leq \calpha n/s$,  for every $v \in V$. 
In our case we set $s = n^{2/3 + \alpha}$, thus $|A|=\tilde{O}(n^{2/3 + \alpha})$ and $|B(v)|, |C(v)| \leq \calpha  n^{1/3-\alpha}$, for every $v \in V$. 

For every vertex $v \in V$, the ADO explicitly stores the distances between $v$ and every vertex in $C(N(v, \ccalphabeta  n^{1/3 + 2\alpha}))$ for some constant $\cfp$ to be decided later. In addition, for every vertex $v \in V$ the ADO stores $p(v)$, $d(v, p(v))$ and the distances between $v$ and every vertex in $A$. 

A distance query between vertices $u$ and $v$ is answered as follows. If one of the following conditions holds 
\begin{enumerate*} 
\item $u \in A$ or $v \in A$, 
\item $u \in C(N(v, \ccalphabeta  n^{1/3 + 2\alpha}))$  or $v \in C(N(u, \ccalphabeta  n^{1/3 + 2\alpha}))$
\end{enumerate*}, then the exact distance is returned. Otherwise, the ADO returns $\hat{d}(u,v) = \min\{ d(u,p(u)) + d(p(u), v), \, d(u,p(v)) + d(p(v), v) \}$. Notice that the query time is constant.

In Claim \ref{clm:new ADO, space}, we show that the space complexity of the ADO is $\tilde{O}(\cfp n^{5/3+\alpha})$ and in Claim \ref{clm:stretch analysis}, we show that the ADO satisfies a $(2, 1- rad(\cspace n^{3\alpha}))$-stretch.

\begin{claim}\label{clm:new ADO, space}
The space complexity of the ADO is $\tilde{O}(\cfp n^{5/3+\alpha})$.
\end{claim} 

\begin{proof}
Storing $p(v)$, $d(v, p(v))$ and the distances between $v$ and every vertex in $A$, for all vertices $v \in V$, uses $\tilde{O}(n \cdot n^{2/3+\alpha}) = \tilde{O}(n^{5/3 + \alpha})$ space. As mentioned in the construction phase, $|B(v)|, |C(v)| \leq O(n^{1/3-\alpha})$ for every $v \in V$. 
Thus, storing the distances between every vertex $v$ and $C(N(v, \ccalphabeta  n^{1/3 + 2\alpha}))$ requires $O(n \cdot n^{1/3-\alpha} \cdot \ccalphabeta n^{1/3+2\alpha}) = \tilde{O}(\cfp n^{5/3 + \alpha})$ space as well, leading to an overall space complexity of $\tilde{O}(\cfp n^{5/3 + \alpha})$. \end{proof}

\begin{claim}\label{clm:stretch analysis}
The distance estimation $\hat{d}(u,v)$ returned by the ADO satisfies $d(u,v) \leq \hat{d}(u,v) \leq 2 d(u,v) + 1 - rad(\cspace n^{3\alpha})$.
\end{claim} \begin{proof}
Notice that $d(u,v) \leq \hat{d}(u,v)$ since the ADO always returns a length of some path in the graph between $u$ and $v$. It is left to show that $\hat{d}(u,v) \leq 2 d(u,v) + 1 - rad(\cspace n^{3\alpha})$. 

If the exact distance is stored in the ADO then $\hat{d}(u,v) = d(u,v)$ and the claim follows. 
Consider the case that the exact distance is not stored. This implies that $u, v\notin A$ and $v \notin C(N(u, \ccalphabeta  n^{1/3 + 2\alpha}))$. 
Assume towards a contradiction that $N(u, \ccalphabeta  n^{1/3 + 2\alpha}) \cap B(v)  \neq \emptyset$ and let $w$ be a vertex such that $w \in N(u, \ccalphabeta  n^{1/3 + 2\alpha}) \cap B(v) $. 
From the definitions of bunch and cluster, we have that $w \in B(v)$ if and only if  $v \in C(w)$. 
Thus, $v \in C(w)$, and since $w \in N(u, \ccalphabeta  n^{1/3 + 2\alpha})$, it must be that $v \in C(N(u, \ccalphabeta  n^{1/3 + 2\alpha}))$ which is a contradiction.
Thus, we have that $N(u, \ccalphabeta  n^{1/3 + 2\alpha}) \cap B(v)  = \emptyset$.

\tikzset{every picture/.style={line width=0.75pt}}

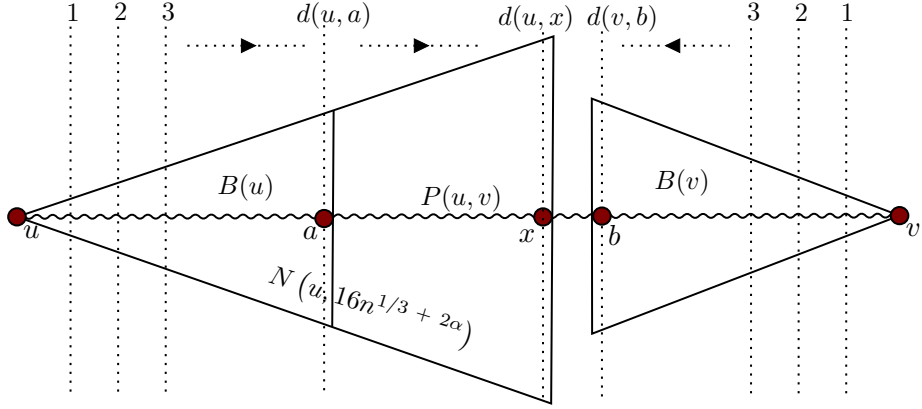
\begin{figure}[t]
\tikzset{every picture/.style={line width=0.75pt}} 

\begin{center}

\tikzset{every picture/.style={line width=0.75pt}} 

\begin{tikzpicture}[x=0.75pt,y=0.75pt,yscale=-0.8,xscale=0.8]

\draw [line width=0.75]    (48.6,149.98) .. controls (50.27,148.31) and (51.93,148.3) .. (53.6,149.97) .. controls (55.27,151.64) and (56.93,151.64) .. (58.6,149.97) .. controls (60.27,148.3) and (61.93,148.3) .. (63.6,149.96) .. controls (65.27,151.62) and (66.93,151.62) .. (68.6,149.95) .. controls (70.27,148.28) and (71.93,148.28) .. (73.6,149.94) .. controls (75.27,151.61) and (76.93,151.61) .. (78.6,149.94) .. controls (80.27,148.27) and (81.93,148.27) .. (83.6,149.93) .. controls (85.27,151.59) and (86.93,151.59) .. (88.6,149.92) .. controls (90.27,148.25) and (91.93,148.25) .. (93.6,149.92) .. controls (95.27,151.58) and (96.93,151.58) .. (98.6,149.91) .. controls (100.27,148.24) and (101.93,148.24) .. (103.6,149.9) .. controls (105.27,151.57) and (106.93,151.57) .. (108.6,149.9) .. controls (110.27,148.23) and (111.93,148.23) .. (113.6,149.89) .. controls (115.27,151.55) and (116.93,151.55) .. (118.6,149.88) .. controls (120.27,148.21) and (121.93,148.21) .. (123.6,149.87) .. controls (125.27,151.54) and (126.93,151.54) .. (128.6,149.87) .. controls (130.27,148.2) and (131.93,148.2) .. (133.6,149.86) .. controls (135.27,151.52) and (136.93,151.52) .. (138.6,149.85) .. controls (140.27,148.18) and (141.93,148.18) .. (143.6,149.85) .. controls (145.27,151.51) and (146.93,151.51) .. (148.6,149.84) .. controls (150.27,148.17) and (151.93,148.17) .. (153.6,149.83) .. controls (155.27,151.5) and (156.93,151.5) .. (158.6,149.83) .. controls (160.27,148.16) and (161.93,148.16) .. (163.6,149.82) .. controls (165.27,151.48) and (166.93,151.48) .. (168.6,149.81) .. controls (170.27,148.14) and (171.93,148.14) .. (173.6,149.81) .. controls (175.27,151.47) and (176.93,151.47) .. (178.6,149.8) .. controls (180.27,148.13) and (181.93,148.13) .. (183.6,149.79) .. controls (185.27,151.45) and (186.93,151.45) .. (188.6,149.78) .. controls (190.27,148.11) and (191.93,148.11) .. (193.6,149.78) .. controls (195.27,151.44) and (196.93,151.44) .. (198.6,149.77) .. controls (200.27,148.1) and (201.93,148.1) .. (203.6,149.76) .. controls (205.27,151.43) and (206.93,151.43) .. (208.6,149.76) .. controls (210.27,148.09) and (211.93,148.09) .. (213.6,149.75) .. controls (215.27,151.41) and (216.93,151.41) .. (218.6,149.74) .. controls (220.27,148.07) and (221.93,148.07) .. (223.6,149.74) .. controls (225.27,151.4) and (226.93,151.4) .. (228.6,149.73) .. controls (230.27,148.06) and (231.93,148.06) .. (233.6,149.72) .. controls (235.27,151.38) and (236.93,151.38) .. (238.6,149.71) .. controls (240.27,148.04) and (241.93,148.04) .. (243.6,149.71) .. controls (245.27,151.37) and (246.93,151.37) .. (248.6,149.7) .. controls (250.27,148.03) and (251.93,148.03) .. (253.6,149.69) .. controls (255.27,151.36) and (256.93,151.36) .. (258.6,149.69) .. controls (260.27,148.02) and (261.93,148.02) .. (263.6,149.68) .. controls (265.27,151.34) and (266.93,151.34) .. (268.6,149.67) .. controls (270.27,148) and (271.93,148) .. (273.6,149.67) .. controls (275.27,151.33) and (276.93,151.33) .. (278.6,149.66) .. controls (280.27,147.99) and (281.93,147.99) .. (283.6,149.65) .. controls (285.27,151.32) and (286.93,151.32) .. (288.6,149.65) .. controls (290.27,147.98) and (291.93,147.98) .. (293.6,149.64) .. controls (295.27,151.3) and (296.93,151.3) .. (298.6,149.63) .. controls (300.27,147.96) and (301.93,147.96) .. (303.6,149.62) .. controls (305.27,151.29) and (306.93,151.29) .. (308.6,149.62) .. controls (310.27,147.95) and (311.93,147.95) .. (313.6,149.61) .. controls (315.27,151.27) and (316.93,151.27) .. (318.6,149.6) .. controls (320.27,147.93) and (321.93,147.93) .. (323.6,149.6) .. controls (325.27,151.26) and (326.93,151.26) .. (328.6,149.59) .. controls (330.27,147.92) and (331.93,147.92) .. (333.6,149.58) .. controls (335.27,151.25) and (336.93,151.25) .. (338.6,149.58) .. controls (340.27,147.91) and (341.93,147.91) .. (343.6,149.57) .. controls (345.27,151.23) and (346.93,151.23) .. (348.6,149.56) .. controls (350.27,147.89) and (351.93,147.89) .. (353.6,149.55) .. controls (355.27,151.22) and (356.93,151.22) .. (358.6,149.55) .. controls (360.27,147.88) and (361.93,147.88) .. (363.6,149.54) .. controls (365.27,151.2) and (366.93,151.2) .. (368.6,149.53) .. controls (370.27,147.86) and (371.93,147.86) .. (373.6,149.53) .. controls (375.27,151.19) and (376.93,151.19) .. (378.6,149.52) .. controls (380.27,147.85) and (381.93,147.85) .. (383.6,149.51) .. controls (385.27,151.18) and (386.93,151.18) .. (388.6,149.51) .. controls (390.27,147.84) and (391.93,147.84) .. (393.6,149.5) .. controls (395.27,151.16) and (396.93,151.16) .. (398.6,149.49) .. controls (400.27,147.82) and (401.93,147.82) .. (403.6,149.49) .. controls (405.27,151.15) and (406.93,151.15) .. (408.6,149.48) .. controls (410.27,147.81) and (411.93,147.81) .. (413.6,149.47) .. controls (415.27,151.13) and (416.93,151.13) .. (418.6,149.46) .. controls (420.27,147.79) and (421.93,147.79) .. (423.6,149.46) .. controls (425.27,151.12) and (426.93,151.12) .. (428.6,149.45) .. controls (430.27,147.78) and (431.93,147.78) .. (433.6,149.44) .. controls (435.27,151.11) and (436.93,151.11) .. (438.6,149.44) .. controls (440.27,147.77) and (441.93,147.77) .. (443.6,149.43) .. controls (445.27,151.09) and (446.93,151.09) .. (448.6,149.42) .. controls (450.27,147.75) and (451.93,147.75) .. (453.6,149.42) .. controls (455.27,151.08) and (456.93,151.08) .. (458.6,149.41) .. controls (460.27,147.74) and (461.93,147.74) .. (463.6,149.4) .. controls (465.27,151.07) and (466.93,151.07) .. (468.6,149.4) .. controls (470.27,147.73) and (471.93,147.73) .. (473.6,149.39) .. controls (475.27,151.05) and (476.93,151.05) .. (478.6,149.38) .. controls (480.27,147.71) and (481.93,147.71) .. (483.6,149.37) .. controls (485.27,151.04) and (486.93,151.04) .. (488.6,149.37) .. controls (490.27,147.7) and (491.93,147.7) .. (493.6,149.36) .. controls (495.27,151.02) and (496.93,151.02) .. (498.6,149.35) .. controls (500.27,147.68) and (501.93,147.68) .. (503.6,149.35) .. controls (505.27,151.01) and (506.93,151.01) .. (508.6,149.34) .. controls (510.27,147.67) and (511.93,147.67) .. (513.6,149.33) .. controls (515.27,151) and (516.93,151) .. (518.6,149.33) .. controls (520.27,147.66) and (521.93,147.66) .. (523.6,149.32) .. controls (525.27,150.98) and (526.93,150.98) .. (528.6,149.31) .. controls (530.27,147.64) and (531.93,147.64) .. (533.6,149.3) .. controls (535.27,150.97) and (536.93,150.97) .. (538.6,149.3) .. controls (540.27,147.63) and (541.93,147.63) .. (543.6,149.29) .. controls (545.27,150.95) and (546.93,150.95) .. (548.6,149.28) .. controls (550.27,147.61) and (551.93,147.61) .. (553.6,149.28) .. controls (555.27,150.94) and (556.93,150.94) .. (558.6,149.27) .. controls (560.27,147.6) and (561.93,147.6) .. (563.6,149.26) .. controls (565.27,150.93) and (566.93,150.93) .. (568.6,149.26) .. controls (570.27,147.59) and (571.93,147.59) .. (573.6,149.25) .. controls (575.27,150.91) and (576.93,150.91) .. (578.6,149.24) .. controls (580.27,147.57) and (581.93,147.57) .. (583.6,149.24) .. controls (585.27,150.9) and (586.93,150.9) .. (588.6,149.23) .. controls (590.27,147.56) and (591.93,147.56) .. (593.6,149.22) .. controls (595.27,150.88) and (596.93,150.88) .. (598.6,149.21) .. controls (600.27,147.54) and (601.93,147.54) .. (603.6,149.21) -- (604.99,149.21) -- (604.99,149.21) ;
\draw   (48.6,149.98) -- (386.8,36.1) -- (385.35,267.75) -- cycle ;
\draw   (604.99,149.21) -- (411.17,223.7) -- (410.97,75.38) -- cycle ;
\draw  [fill={rgb, 255:red, 130; green, 0; blue, 0 }  ,fill opacity=1 ] (43.04,149.98) .. controls (43.04,146.97) and (45.53,144.53) .. (48.6,144.53) .. controls (51.67,144.53) and (54.16,146.97) .. (54.16,149.98) .. controls (54.16,152.99) and (51.67,155.43) .. (48.6,155.43) .. controls (45.53,155.43) and (43.04,152.99) .. (43.04,149.98) -- cycle ;
\draw    (248.22,82.74) -- (247.33,219.33) ;
\draw [line width=0.75]  [dash pattern={on 0.84pt off 2.51pt}]  (82.36,29.95) -- (82.36,262.95) ;
\draw [line width=0.75]  [dash pattern={on 0.84pt off 2.51pt}]  (112.36,29.95) -- (112.36,262.95) ;
\draw [line width=0.75]  [dash pattern={on 0.84pt off 2.51pt}]  (142.36,29.95) -- (142.36,262.95) ;
\draw [line width=0.75]  [dash pattern={on 0.84pt off 2.51pt}]  (511.36,31) -- (511.36,264) ;
\draw [line width=0.75]  [dash pattern={on 0.84pt off 2.51pt}]  (541.36,31) -- (541.36,264) ;
\draw [line width=0.75]  [dash pattern={on 0.84pt off 2.51pt}]  (571.36,31) -- (571.36,264) ;
\draw [line width=0.75]  [dash pattern={on 0.84pt off 2.51pt}]  (242.22,29.09) -- (242.22,262.09) ;
\draw [line width=0.75]  [dash pattern={on 0.84pt off 2.51pt}]  (417.36,33) -- (417.36,266) ;
\draw  [fill={rgb, 255:red, 130; green, 0; blue, 0 }  ,fill opacity=1 ] (236.66,151.04) .. controls (236.66,148.03) and (239.15,145.59) .. (242.22,145.59) .. controls (245.29,145.59) and (247.78,148.03) .. (247.78,151.04) .. controls (247.78,154.05) and (245.29,156.49) .. (242.22,156.49) .. controls (239.15,156.49) and (236.66,154.05) .. (236.66,151.04) -- cycle ;
\draw  [fill={rgb, 255:red, 130; green, 0; blue, 0 }  ,fill opacity=1 ] (599.43,149.21) .. controls (599.43,146.2) and (601.92,143.76) .. (604.99,143.76) .. controls (608.06,143.76) and (610.55,146.2) .. (610.55,149.21) .. controls (610.55,152.22) and (608.06,154.65) .. (604.99,154.65) .. controls (601.92,154.65) and (599.43,152.22) .. (599.43,149.21) -- cycle ;
\draw  [fill={rgb, 255:red, 130; green, 0; blue, 0 }  ,fill opacity=1 ] (411.8,149.5) .. controls (411.8,146.49) and (414.29,144.05) .. (417.36,144.05) .. controls (420.43,144.05) and (422.92,146.49) .. (422.92,149.5) .. controls (422.92,152.51) and (420.43,154.95) .. (417.36,154.95) .. controls (414.29,154.95) and (411.8,152.51) .. (411.8,149.5) -- cycle ;
\draw  [fill={rgb, 255:red, 130; green, 0; blue, 0 }  ,fill opacity=1 ] (374.66,150.04) .. controls (374.66,147.03) and (377.15,144.59) .. (380.22,144.59) .. controls (383.29,144.59) and (385.78,147.03) .. (385.78,150.04) .. controls (385.78,153.05) and (383.29,155.49) .. (380.22,155.49) .. controls (377.15,155.49) and (374.66,153.05) .. (374.66,150.04) -- cycle ;
\draw [line width=0.75]  [dash pattern={on 0.84pt off 2.51pt}]  (380.22,33.54) -- (380.22,266.54) ;
\draw [line width=0.75]  [dash pattern={on 0.84pt off 2.51pt}]  (430,42) -- (498,42) ;
\draw [shift={(457.5,42)}, rotate = 0] [fill={rgb, 255:red, 0; green, 0; blue, 0 }  ][line width=0.08]  [draw opacity=0] (8.93,-4.29) -- (0,0) -- (8.93,4.29) -- cycle    ;
\draw [line width=0.75]  [dash pattern={on 0.84pt off 2.51pt}]  (156,42) -- (232,42) ;
\draw [shift={(199,42)}, rotate = 180] [fill={rgb, 255:red, 0; green, 0; blue, 0 }  ][line width=0.08]  [draw opacity=0] (8.93,-4.29) -- (0,0) -- (8.93,4.29) -- cycle    ;
\draw [line width=0.75]  [dash pattern={on 0.84pt off 2.51pt}]  (265,42) -- (341,42) ;
\draw [shift={(308,42)}, rotate = 180] [fill={rgb, 255:red, 0; green, 0; blue, 0 }  ][line width=0.08]  [draw opacity=0] (8.93,-4.29) -- (0,0) -- (8.93,4.29) -- cycle    ;

\draw (227.02,154.28) node [anchor=north west][inner sep=0.75pt]    {$a$};
\draw (419.36,152.9) node [anchor=north west][inner sep=0.75pt]    {$b$};
\draw (50.6,153.38) node [anchor=north west][inner sep=0.75pt]  [font=\normalsize]  {$u$};
\draw (606.99,152.61) node [anchor=north west][inner sep=0.75pt]  [font=\normalsize]  {$v$};
\draw (173.04,121.38) node [anchor=north west][inner sep=0.75pt]  [font=\small]  {$B( u)$};
\draw (448.57,117.78) node [anchor=north west][inner sep=0.75pt]  [font=\small]  {$B( v)$};
\draw (209.69,174.5) node [anchor=north west][inner sep=0.75pt]  [font=\small,rotate=-17.08]  {$N\left( u,16n^{1/3\ +\ 2\alpha }\right)$};
\draw (78,13.81) node [anchor=north west][inner sep=0.75pt]  [font=\small]  {$1$};
\draw (223,12.4) node [anchor=north west][inner sep=0.75pt]  [font=\small]  {$d( u,a)$};
\draw (406,14.4) node [anchor=north west][inner sep=0.75pt]  [font=\small]  {$d( v,b)$};
\draw (108,13.4) node [anchor=north west][inner sep=0.75pt]  [font=\small]  {$2$};
\draw (138,13.4) node [anchor=north west][inner sep=0.75pt]  [font=\small]  {$3$};
\draw (507,13.4) node [anchor=north west][inner sep=0.75pt]  [font=\small]  {$3$};
\draw (537,13.99) node [anchor=north west][inner sep=0.75pt]  [font=\small]  {$2$};
\draw (566,13.99) node [anchor=north west][inner sep=0.75pt]  [font=\small]  {$1$};
\draw (363.02,156.28) node [anchor=north west][inner sep=0.75pt]    {$x$};
\draw (350,14.4) node [anchor=north west][inner sep=0.75pt]  [font=\small]  {$d( u,x)$};
\draw (301,128.4) node [anchor=north west][inner sep=0.75pt]  [font=\small]  {${\displaystyle P( u,v)}$};

\end{tikzpicture}

\end{center}

\captionof{figure}{A query for $u, v \in V$ in the case that $N(u, \ccalphabeta  n^{1/3 + 2\alpha}) \cap B(v)  = \emptyset$. Since $x \in N(u, \ccalphabeta  n^{1/3 + 2\alpha})$, $b \in B(v)$ and $x$ and $b$ are both on a shortest path between $u$ and $v$, it must be that $d(u,x) \leq d(u,b) - 1$.
}
\vspace{0.2cm}
\label{fig:ADO_main}
\end{figure}

Let $P(u,v)$ be a shortest path between $u$ and $v$. Let $a$ be the furthest vertex from $u$ in $B(u) \cap P(u,v)$, let $x$ be the furthest vertex from $u$ in $N(u, \ccalphabeta  n^{1/3 + 2\alpha}) \cap P(u,v)$ and let $b$ be the furthest vertex from $v$ in $B(v) \cap P(u,v)$ (see Figure~\ref{fig:ADO_main}). 
Notice that, by definition of $x$ and $ecc(v, s)$, if $T(u, d(u, x)) \subseteq N(u, \ccalphabeta  n^{1/3 + 2\alpha})$ then $d(u,x) = ecc(u, \ccalphabeta  n^{1/3 + 2\alpha})$  and if $T(u, d(u, x)) \not\subseteq N(u, \ccalphabeta  n^{1/3 + 2\alpha})$ then $d(u,x) = ecc(u, \ccalphabeta  n^{1/3 + 2\alpha}) + 1$. Thus, we get that $d(u,x) \geq ecc(u, \ccalphabeta  n^{1/3 + 2\alpha})$.

Since $N(u, \ccalphabeta  n^{1/3 + 2\alpha}) \cap B(v)  = \emptyset$, $x \in N(u, \ccalphabeta  n^{1/3 + 2\alpha})$, $b \in B(v)$ and $x$ and $b$ are both on a shortest path between $u$ and $v$, it must be that $d(u,x) \leq d(u,b) - 1$. Since $b$ is on a shortest path between $u$ and $v$, it holds that $d(u,b) = d(u,v) - d(v,b)$, and so $d(u,x) \leq d(u,v) - d(v,b) - 1$. Since $d(u,x) \geq ecc(u, \ccalphabeta  n^{1/3 + 2\alpha})$, it follows that: 
\begin{align}ecc(u, \ccalphabeta  n^{1/3 + 2\alpha}) \leq d(u,v) - d(v,b) - 1.\label{eq:2}\end{align} 
By Lemma \ref{lem:2}, $ecc(u, \ceil*{\calpha n^{1/3 - \alpha}}) + rad(\ceil*{\cspace n^{3 \alpha}}) \leq ecc(u, (\calpha n^{1/3 - \alpha}+1)(\cspace n^{3 \alpha}+2))$. Since $a \in B(u)$ and $|B(u)| \leq \calpha n^{1/3-\alpha}$, it follows from the definitions of $ecc(v, s)$ and bunch that $d(u, a) \leq ecc(u, \calpha n^{1/3 - \alpha})$. 
We have that:
\begin{align*}
d(u, a) + rad(\cspace n^{3 \alpha}) &\leq ecc(u, \calpha n^{1/3 - \alpha}) + rad(\cspace n^{3 \alpha})\\
&\leq ecc(u, \ceil*{\calpha n^{1/3 - \alpha}}) + rad(\ceil*{\cspace n^{3 \alpha}})\\
& \leq ecc(u, (\calpha n^{1/3 - \alpha}+1)(\cspace n^{3 \alpha}+2)) \\
&\leq ecc(u, \ccalphabeta  n^{1/3 + 2\alpha})\\
&\leqexpl{Equation \eqref{eq:2}} d(u,v) - d(v,b) - 1.
\end{align*}

Thus, $d(u, a) + d(b, v) \leq d(u,v) - rad(\cspace n^{3 \alpha}) - 1$. It follows that: \begin{align}
2 \nocdot \min\{d(u,a), d(b,v)\} \leq d(u,v) - rad(\cspace n^{3\alpha}) - 1.\label{eq:3}\end{align} Notice that by the definitions of bunch, $a$ and $p(u)$, it holds that $d(u,p(u)) = d(u,a) + 1$. Similarly, $d(v,p(v)) = d(v,b) + 1$. Thus: 
\begin{align}\hat{d}(u,v) &\leq \min\{d(u,p(u)) + d(p(u), v), d(u,p(v)) + d(p(v), v)\} \notag\\
&\leqexpl{triangle inequallity} \min\{2d(u,p(u)) + d(u,v), 2d(v,p(v)) + d(u,v)\} \notag\\
&\equalexpl{$d(u,p(u)) = d(u,a) + 1$ and $d(v,p(v)) = d(v,b) + 1$} \min\{2(d(u,a)+1)+d(u,v), 2(d(v,b)+1) + d(u,v)\} \notag\\
&\leq 2\min\{d(u,a), d(v,b)\} +2+d(u,v) \label{eq:4}\\
&\leqexpl{Equation \eqref{eq:3}} d(u,v) - rad(\cspace n^{3\alpha}) + 1 + d(u,v) \notag\\
&= 2d(u,v) + 1 - rad(\cspace n^{3\alpha}). \label{eq:5} \end{align} 
\end{proof}

By combining our ADO construction with Claims~\ref{clm:new ADO, space} and~\ref{clm:stretch analysis} we have proven the following lemma.
\begin{lemma}\label{tho:main_1}
For any graph $G$ with $n$ vertices, real $0 \leq \alpha < \frac{1}{3}$ and constant $\cfp \geq 1$, it is possible to construct an ADO that uses $\tilde{O}(\cfp n^{\frac{5}{3} + \alpha})$ space and has a $(2, 1- rad(\cspace n^{3\alpha}))$-stretch.
\end{lemma}

\subsection{Proof of Main Upper Bound Theorem}

The following lemma connects $\Delta_G$ and $rad(s)$, which is the last ingredient needed for proving Theorem~\ref{thm:main_degree}.

\begin{lemma}\label{lem:5}
Let $G=(V, E)$ be an unweighted undirected graph, with $|V| = n$. For any real $s$ such that $1 \leq s < n$, it holds that $rad_G(s) \geq \lfloor \log_{\Delta_G} (s/2) \rfloor$.
\end{lemma}

\begin{proof} 
For any vertex $v$ and integer $t \geq 1$, $T(v, t)$ cannot include more than $\Delta_G \cdot \sum_{i=0}^{t-1} (\Delta_G-1)^{i}$ vertices. 
Since $\Delta_G, t \geq 1$ we have that $\Delta_G \cdot \sum_{i=0}^{t-1} (\Delta_G-1)^{i} \leq 2 \cdot \Delta_G^t$ and so for any integer $t \geq 1$ such that $2 \cdot \Delta_G^t < n$ it must be that $T(v, t) \subseteq N(v, 2 \cdot \Delta_G^t)$.
By definition, $ecc(v, s)$ is equal to the largest integer $x \in \left[0, ecc(v) \right]$ for which $T(v, x) \subseteq N(v, s)$.
Thus, $t \leq ecc(v, 2 \cdot \Delta_G^t)$. Since $t \leq ecc(v, 2 \cdot \Delta_G^t)$ for any vertex $v$, it follows from the definition of $rad(s) = \min_{v \in V}\{ecc(v, s)\}$ that $t \leq rad(2 \cdot \Delta_G^t)$. 
Setting $s \geq 2 \cdot \Delta_G^t$, or $t \leq \log_{\Delta_G} (s /2)$, it follows that for any integer $t$ such that $t \leq \log_{\Delta_G} (s /2)$ it must be that $t \leq rad(2 \cdot \Delta_G^t) \leq rad(s)$. Thus, $\lfloor \log_{\Delta_G} (s /2) \rfloor \leq rad(s)$.
\end{proof}

Finally, we are ready to prove Theorem~\ref{thm:main_degree}.

\begin{proof}[Proof of \Cref{thm:main_degree}]
By Lemma \ref{lem:5},  $k = \lfloor \log_{\Delta_G} (\Delta_G^k) \rfloor \leq \lfloor \log_{\Delta_G} (c^k n^{1 - k \varepsilon}) \rfloor \leq rad_G(2 c^k n^{1 - k \varepsilon})$. Thus, the ADO from Lemma \ref{tho:main_1} constructed for $G$ using $\alpha = \frac{1 - k \varepsilon}{3}$ and $\cfp  = 2 c^k$ uses $\tilde{O}(c^k n^{2-\frac{k \varepsilon}{3}})$ space and produces a distance estimation that satisfies $d(u,v) \leq \hat{d}(u,v) \leq \max \{d(u,v), \, 2d(u,v) + 1 - rad_G(2 c^k n^{1 - k \varepsilon})\} \leq \max \{d(u,v), \, 2d(u,v) + 1 - k\}$.
\end{proof}

\section{Conditional Lower Bounds}\label{sec: lower bound proof}

\subsection{Conditional Lower Bound on Additive Error Improvement} 
As mentioned in Section~\ref{sec:intro}, proving \Cref{lem:CLB general distinguisher}, which assuming Hypothesis \ref{set intersection hypothesis}, eliminates the possibility of a subquadratic $(k,k+2)$-distinguisher oracle for graphs with $\Delta_G = \Theta(n^{1/k})$, directly implies Theorem~\ref{thm:main lower bound}, since a $(2, 1-k)$-stretch ADO is also a $(k,k+2)$-distinguisher oracle.

\begin{proof}[Proof of \Cref{lem:CLB general distinguisher}.]
Given an instance of Problem~\ref{set intersection problem}, we construct a graph $G$ with $n = \tilde{O}(N)$ vertices and $\Delta_G = O(n^{1/k})$, such that a $(k, k+2)$-distinguisher oracle for $G$ solves the instance of Problem \ref{set intersection problem}.

We begin by focusing on a $k+1$ layered graph $\mathcal{L} = (V_\mathcal{L},E_\mathcal{L})$, which we call the \emph{infrastructure graph}. 
The infrastructure graph has three important properties:
\begin{enumerate*}
    \item each layer contains $N$ vertices, 
    \item there is a path of length $k$ from every vertex in the first layer to every vertex in the last layer, and 
    \item the degree of every vertex is at most $2N^{1/k}$.
\end{enumerate*}

We then construct
for each $x\in X$ a graph $G_x = (V_x,E_x)$,
which is a subgraph of (a copy of) $\mathcal L$, by removing some of the edges between the first (last) and second (second to last) layers of $\mathcal L$ in a way that expresses which sets contain $x$ and which do not. 
Finally, we construct the graph $G$ which is \emph{specialized}  union of all of the graphs $G_x$ for all $x\in X$, and enables solving the instance of Problem~\ref{set intersection problem} by using a $(k,k+2)$-distinguisher oracle on $G$.

\paragraph{The infrastructure graph.}
The infrastructure graph $\mathcal{L}$ is a $k+1$ layered graph where each layer contains $N$ vertices, and each layer of $N$ vertices is locally indexed from $1$ to $N$. The layers are numbered $0$ to $k$.

The edges of $\mathcal{L}$ are defined using the following labels. 
Assign a \emph{label} $\ell(v)$  to every vertex $v$ in $\mathcal L$  which is the  $k$ digit representation in base\footnote{We assume for convenience that $N^{1/k}$ is an integer, since otherwise, one can increase $N$ slightly without affecting the asymptotic complexities.} $N^{1/k}$ of the local index (an integer between $1$ and $N$) of $v$.
Then, for every $1\le t \le k$, connect $u$ from layer $t-1$ with $v$ from layer $t$ if and only if the digits of $\ell(u)$ and the digits of $\ell(v)$ all match, except for possibly the $t$'th digit. 
It is straightforward to observe (since each digit has $N^{1/k}$ options) that the degree of every vertex in $\mathcal{L}$ is $2N^{1/k}$, except for the vertices in the first and last layers which have degree $N^{1/k}$.
The following claim shows that there is a path of length $k$ from every vertex in the first layer and every vertex in the last layer. 

\begin{claim}\label{clm:layers}
    Let $v$ be a vertex in the first layer of $\mathcal L$ and let $u$ be a vertex in the last layer of $\mathcal L$. 
    Then there exists a path of length $k$ from $u$ to $v$ in $\mathcal L$.
\end{claim}
\begin{proof} 

We describe the path of length $k$ between $u$ and $v$.
For any $0 \leq t \leq k$, consider the vertex $w_t$ in layer $t$ of $\mathcal L$  with the label of the following form: the first $t$ digits are the first $t$ digits of $\ell(u)$, and the last $k-t$ digits are the last $k-t$ digits of $\ell(v)$.  
Thus, for $0\le t \le k-1$,  the edge $(w_t,w_{t+1})$ is in $\mathcal L$ since $\ell(w_t)$ and $\ell(w_{t+1})$ are the same, except for possibly  the $(t+1)$-th digit. 
The set of edges which we described form a path of length $k$ between $v$ and $u$.
\end{proof}

\paragraph{Constructing $G_x$.}
We construct $V_x$ by making copies of each vertex in $V_\mathcal L$.
Denote the first layer of $\mathcal L$ by $V_L = \{v_1, \dots, v_N\}$ and the last layer by $V_R = \{u_1, \dots, u_N\}$. 
Let $\hat E_x = \{(v_i,w)| x\notin S_i \wedge (v_i,w)\in E_\mathcal{L}\} \cup \{(u_i,w)| x\notin S_i \wedge (u_i,w)\in E_\mathcal{L}\}$.
Thus, $\hat E_x$ is the set of edges in $\mathcal L$ that touch vertices in the first or last layers whose index corresponds to the index of sets that do not contain $x$.
We construct $E_x$ by making copies of all edges in $E_{\mathcal L}\setminus \hat E_x$.
The reason for removing the edges in $\hat E_x$ is so that vertices   in the first and last layers of $G_x$ whose edges are in $\hat E_x$ are not connected to any other vertex in $G_x$.
Thus, for each $v_i$ ($u_j$) in the first (last) layer of $G_x$,   $x\in S_i$ if and only if there are edges between $v_i$ ($u_j$) and the second (second to last) layer in $G_x$. 
By \Cref{clm:layers}, if $S_i\cap S_j \ne \emptyset$ then there exists a path of length $k$ between $v_i$ and $u_j$, and otherwise, there is no path in $G_x$ between $v_i$ and $u_j$.
Finally, since $G_x$ is a partial copy of $\mathcal{L}$, the maximum degree in $G_x$ is $2N^{1/k}$.

\paragraph{Constructing $G$.}
We construct the $k+1$ layered graph $G$ by performing the following special union of $G_x$ for all $x$:
for $1\le t \le k-1$ the $t$'th layer of $G$ is the union of the $t$'th layer of all of the $G_x$ graphs taken over all $x\in X$.
Thus, each of the $k-1$ inner layers (excluding the first and last layer of $G$) has $|X|N$ vertices.
For the first (last) layer $G$, instead of taking the union of all of the first (last) layers from all of the $G_x$ graphs, we merge them all into one layer of $N$ vertices. 
So the $i$'th vertex in the first (last) layer of $G$ is a vertex obtained by merging the $i$'th vertex in the first (last) layer of every $G_x$, for all $x\in X$.
Thus, the first and last layers of $G_x$ contain $N$ vertices each.
Since the vertices in the first and last layer of $G$ correspond directly to the vertices $V_L$ and $V_R$ in $\mathcal L$, respectively, we treat the first layer of $G$ as $V_L=\{v_1,\ldots v_n\}$ and the last layer of $G$ by $V_R = \{u_1,\ldots, u_N\}$.
Thus, each node in $V_L\cup V_R$ has  maximum degree at most $|X|N^{1/k} = \tilde O(N^{1/k}).$

\paragraph{Answering a set intersection query.}
Notice that for a set intersection query between $S_i$ and $S_j$, if $S_i\cap S_j \neq \emptyset$, then there exists some $x\in S_i\cap S_j$, and since $G$ contains $G_x$ as a subgraph, the distance between $v_i$ and $u_j$ is at most (and actually exactly) $k$.  
On the other hand, if there exists a path $P$ of length $k$ between  $v_i$ and $u_j$, then by the construction of $G$, $P$ must be completely contained within some $G_x$ for some $x\in X$. 
By the construction of $G_x$, and specifically $E_x$, the existence of $P$ in $G_x$ implies that $x\in S_i$ and $x\in S_j$.
So, in such a case $S_i\cap S_j \neq \emptyset$. 

Notice that, since $G$ is a $k+1$ layered graph, any path between a vertex in the first layer and a vertex in the last layer must be of length $k+2q$ for some integer $q\ge 0$.
Thus, to answer a set intersection query, it suffices to establish whether the distance in $G$ between  $v_i$ and $u_j$ is either $k$ or at least $k+2$, which the $(k,k+2)$-distinguisher oracle returns in constant time. 

\paragraph{Analysis.}
We conclude that a $(k, k+2)$-distinguisher oracle for graphs with $n = \tilde{O}(N)$ vertices and maximum degree $\tilde{O}(n^{1/k})$  also solves the instance of Problem \ref{set intersection problem} (of size $N$). 
Thus, according to Hypothesis \ref{set intersection hypothesis}, an ADO for graphs with $n = \tilde{O}(N)$ vertices, for which the maximum degree is $\tilde{O}(n^{1/k})$, must use $\tilde{O}(N^2) = \tilde{O}(n^2)$ space.
We note that the maximum degree can be reduced to $O(n^{1/k})$ by artificially adding $\tilde{O}(n) = \tilde{O}(N)$ isolated vertices to $G$.
\end{proof}

\subsection{Conditional Lower Bound on Multiplicative Error Improvement} 

We now move on to prove \Cref{thm:mult lower bound}, which states that the multiplicative approximation of the ADO from \Cref{thm:main_degree} is optimal under Hypothesis~\ref{set intersection hypothesis}, even when allowing an arbitrarily large constant additive error. We first provide some intuition. Notice that by the way we constructed the graph $G$ in the proof of \Cref{lem:CLB general distinguisher}, a path between representatives $v_i$ and $u_j$ in the case of $S_i \cap S_j = \emptyset$, has to pass through some other set representative vertex in the first or last layer. In order to prove \Cref{thm:mult lower bound}, we want to construct a $(2-\varepsilon, c)$-stretch ADO that can distinguish the case  $S_i \cap S_j = \emptyset$ from the case $S_i \cap S_j \neq \emptyset$. Thus, we want to maximize the ratio between $d(v_i, u_j)$ in the case that $S_i \cap S_j = \emptyset$ to $d(v_i, u_j)$ in the case that $S_i \cap S_j \neq \emptyset$ in $G$. To do so, we split the edges touching the first and last layers of $G$.

\begin{proof}[Proof of \Cref{thm:mult lower bound}.]
We build upon the construction from the proof of \Cref{lem:CLB general distinguisher} but with two changes: 
\begin{enumerate*}
    \item we use $k = \log(n)$ so $\Delta_G = \tilde{O}(N^{1/k}) = \tilde{O}(1)$, and
    \item given constants $\varepsilon, c > 0$ we choose $t=\ceil*{\frac{k+c}{2 \varepsilon}}$ and split each of the edges connecting the first and second (second to last and last) layers of the graph into $t$ edges.
\end{enumerate*}

The number of vertices in the graph is now $O(2N^{1+1/k} + (k-1)|X|N + 2(t-1)|X|N^{1+1/k}) = \tilde{O}(N)$, the number of edges is $O((k-2)|X|N^{1+1/k} + 2t|X|N^{1+1/k}) = \tilde{O}(N)$ and the maximum degree is $O(|X|N^{1/k}) = \tilde{O}(1)$.

For a set intersection query between $S_i$ and $S_j$, if $S_i\cap S_j \neq \emptyset$, then the distance between $v_i$ and $u_j$ is now $2t + k - 2$. On the other hand, if $S_i\cap S_j = \emptyset$, the distance must be $\geq 4t + k - 2$. In the case of $S_i\cap S_j \neq \emptyset$ where the distance is $2t + k - 2$, a $(2 - \varepsilon, c)$-stretch ADO would not return a distance larger than $(2 - \varepsilon)(2t + k - 2) + c$ which by our choice of $t$ is strictly smaller than $4t + k - 2$, which is the smallest possible distance in the case of $S_i\cap S_j = \emptyset$. Thus, a  $(2 - \varepsilon, c)$-stretch ADO could also solve Problem \ref{set intersection problem} (of size $N$) and so assuming Hypothesis~\ref{set intersection hypothesis} such an ADO requires $\tilde{\Omega}(n^2)$.
\end{proof}

\section{Conclusion}\label{sec:conclusions}
In this paper we provide an algorithm (\Cref{thm:main_degree}) and tight conditional lower bounds (\Cref{thm:main lower bound} and \Cref{thm:mult lower bound}) for subquadratic space ADOs as a function of the maximum degree. We show that for graphs with maximum degree $\Delta_G \leq O(n^{1/k-\varepsilon})$, it is possible to construct a subquadratic space ADO with stretch $(2,1-k)$. Furthermore, we show that under Hypothesis~\ref{set intersection hypothesis}, it is impossible to improve the additive approximation of our ADO, nor is it possible to improve the multiplicative approximation to $2-\varepsilon$, even when allowing an arbitrarily large constant additive error.

\bibliographystyle{plainurl}
\bibliography{ado}

\end{document}